\documentclass[journal,twoside,web]{ieeecolor}
\ifdefined\pdfminorversion \pdfminorversion=4 \fi
\usepackage{generic}
\usepackage{cite}
\usepackage{amsmath,amssymb,amsfonts}
\usepackage{algorithmic}
\usepackage{graphicx}
\usepackage{algorithm,algorithmic}
\usepackage{hyperref}
\usepackage{textcomp}

\usepackage{amsthm}

\usepackage{float}
\usepackage{comment}

\newcommand{\R}{\mathbb{R}}

\newcommand{\K}{\mathcal{K}}
\newcommand{\A}{\mathcal{A}}
\newcommand{\J}{\mathcal{J}}
\newcommand{\lb}{\langle}
\newcommand{\rb}{\rangle}
\newcommand{\hstm}{\hspace{2em}}

\theoremstyle{plain}
\newtheorem{thm}{Theorem}
\newtheorem{lem}[thm]{Lemma}
\newtheorem{cor}[thm]{Corollary}
\newtheorem{prop}[thm]{Proposition}
\theoremstyle{definition}
\newtheorem{defn}[thm]{Definition}
\newtheorem{exmp}[thm]{Example}
\newtheorem{prob}[thm]{Problem}

\newtheorem{heur}[thm]{Heuristic}
\newtheorem{assumption}[thm]{Assumption}
\theoremstyle{remark}

\newtheorem*{sktch}{Sketch of Proof}
\newenvironment{sketch}{\begin{sktch}}{\hfill $\qed$ \end{sktch}}

\def\BibTeX{{\rm B\kern-.05em{\sc i\kern-.025em b}\kern-.08em
    T\kern-.1667em\lower.7ex\hbox{E}\kern-.125emX}}
\begin{document}

\title{On Distributed Control of Continuum Swarms: Local Controllers as Differential Operators}

\author{Max Emerick, \IEEEmembership{Graduate Student Member, IEEE}, Saroj Prasad Chhatoi, \IEEEmembership{Graduate Student Member, IEEE}, and Bassam Bamieh, \IEEEmembership{Fellow, IEEE}
\thanks{Submitted for review on August 27, 2026. This work is partially supported by NSF ECCS-2453491 and AFOSR FA9550-23-F-0014.}
\thanks{Max Emerick and Bassam Bamieh are with the Department of Mechanical Engineering, University of California, Santa Barbara, USA (e-mails: \{memerick,bamieh\}@ucsb.edu). Saroj Prasad Chhatoi is with the Institute of Mathematics,   Technische Universität Berlin, Berlin, Germany, (e-mail: chhatoi@math.tu-berlin.de).}
\thanks{The first two authors contributed equally.}
}

\maketitle

\begin{abstract}
We study the problem of distributed control of large-scale robotic swarms which can be modeled as continuum densities evolving under the continuity equation. We propose a formalization of distributed controllers as (generally nonlinear) spatial differential operators, in which control inputs depend only on local information about the state and environment. This perspective yields a fully local, PDE-based framework for analysis and design. We apply this framework to the problem of stabilizing a swarm density around an arbitrary target density, and investigate fundamental limitations of low-spatial-order distributed controllers in achieving this goal. In particular, we show that controllers which act in a purely pointwise manner are incompatible with natural system symmetries and strong forms of stability, and must rely on mixing-type behavior to achieve stabilization. In contrast, we present a simple first-order control law which achieves stabilization and enjoys substantially stronger properties.
\end{abstract}

\begin{IEEEkeywords}
Decentralized control, 
distributed control, 
distributed-parameter systems, 
large-scale systems, 
multi-agent systems, 
partial differential equations
\end{IEEEkeywords}

\section{Introduction}

Robotic swarms are currently being developed for use in a wide variety of applications, including transportation, logistics, data collection, defense, emergency response, and entertainment. These robotic swarms are expected to be very large in some cases, numbering from hundreds to potentially hundreds of thousands of agents.

The analysis and control of such large systems presents significant challenges. For one, the sheer size of these systems makes them difficult to analyze using traditional methods. One approach here is to model these swarms as \emph{continua}, essentially discarding ``microscopic'' information about the states of individual agents to emphasize the ``macroscopic'' configuration of the overall swarm. This provides a tractable and scale-invariant model on which to base analysis and control design. Controllers designed in this continuum setting may then be translated back to the finite-agent setting as discrete approximations.

A second challenge comes from the problem of  control architecture. Traditional \emph{centralized} controllers which observe the state of the whole system and provide commands to each agent require significant communication and computation infrastructure, making these architectures impractical to scale. Thus, very large swarms will need to rely on \emph{distributed} control architectures, where agents observe their surroundings, communicate with their neighbors, and execute their own decisions based on the limited information they have available.

Which distributed control architectures should be implemented is perhaps one of the most important questions of distributed control system co-design. One of the motivations of the present work is that limits-of-performance of various distributed architectures are more transparently studied in the continuum framework. Once those are understood, it is then an easier task to transfer this insight to the discrete, multi-agent setting. While distributed controls have been studied in the discrete-agent setting for some time now, there are relatively few studies on distributed control in the continuum setting. Furthermore, while the formulation of distributed control in the discrete setting is by now well-established, it is not immediately obvious how to formulate the notion in the continuum setting.

Earlier references most closely related to the present work are 
\cite{bamieh2002distributed,d2003distributed,langbort2005distributed,bamieh2005convex,motee2008optimal}, which are primarily concerned with ``architectural questions'' in the design of distributed controllers. In patricular, the interplay between  optimality and the spatial interconnection between measurements and actuation.
The recent work \cite{arbelaiz2024optimal} also explores this idea, but from the estimation side of the problem. In works such as \cite{ferrari2016distributed,krishnan2018distributed,krishnan2025distributed}, authors obtain distributed implementations of centralized control schemes by using distributed algorithms to approximate global quantities. The work \cite{zheng2021distributed} is philosophically similar but focuses on estimation. The recent work \cite{conger2025convex} takes a system-level synthesis approach to control of continuum systems, treating distributed control constraints as constraints on the supports of the kernels of the closed-loop maps.

In this paper, we provide a different notion of distributed controllers as spatial differential operators, and we apply this perspective to a control problem for continuum swarms. While it seems this notion may have been implicit in some of the literature, to our knowledge, it has not been stated formally as such. The main contributions of this work are thus 1) the formal statement of a notion of distributed controllers as differential operators, and 2) new results and perspectives on the fundamental performance limitations of low-spatial-order distributed controllers arising from this notion. We argue that the spatial order of the differential operator in the continuum setting is a proxy for the size of ``local feedback neighborhood'' in the discrete multi-agent setting. 

The rest of the paper proceeds as follows. In Section \ref{problem_formulation_sec}, we state our swarm control problem and develop the notion of distributed controllers as differential operators. In Section \ref{gradient_controller_sec}, we present a simple example of an order-1 distributed controller which solves the proposed problem. In Section \ref{main_results_sec}, we investigate fundamental limitations of order-0 (i.e.~``pointwise'' -- a proxy for fully decentralized controllers) distributed control laws. In particular, we show that order-0 distributed controllers which solve the proposed problem must be based on mixing, and cannot possess certain desirable symmetry or stability properties. We conclude in Section \ref{conclusion_sec} with a brief review and pointers to future work.

\section{Problem Formulation and Description of Model} \label{problem_formulation_sec}

In this section, we state our distributed swarm control problem and develop the notion of distributed controllers as differential operators.

\subsection{Problem Formulation}

We study swarm control problems where the goal is to stabilize a large-scale swarm around some target configuration. The swarm and target are approximated by continuous density functions $\rho$ and $\mu$, respectively, which are assumed to be supported within a compact convex spatial domain $\Omega \subset \R^n$. The dynamics of the swarm density $\rho$ are modeled by the \emph{continuity (transport) equation}
\begin{equation} \label{continuity_eqn}
    \partial_t \rho_t(x) ~=~ - \nabla \cdot \big(\rho_t(x) \, v_t(x)\big), 
    \hstm x\in\Omega\subset\R^n
\end{equation}
equipped with the zero-flux boundary condition $\hat{n} \cdot v_t = 0$ on $\partial \Omega$. We denote time with the subscript $t$ and space with $x$. Both $\rho$ and $v$ are spatio-temporal fields, where $\rho_t(x)\in\R$ is scalar valued and $v_t(x)\in\R^n$ is vector valued. The spatial differential (divergence) operator $\nabla \cdot$ takes a vector field $u$ to a scalar field 
\[
    \big(\nabla \cdot u \big)(x) := \partial_{x_1} u^{(1)}(x) + \cdots +  \partial_{x_n} u^{(n)}(x), 
\]
where $u^{(i)}$ is the $i$'th component of the vector field $u$.

Equation~\eqref{continuity_eqn} describes the transport of a ``density'' $\rho$ by the vector field $v$ under conservation of mass. The vector field $v$ is considered as the control input, while the target density (or ``setpoint'') $\mu$ is assumed to be some fixed (but otherwise arbitrary) density. Note that~\eqref{continuity_eqn} is a bilinear system since the state $\rho_t$ is multiplied by the control input $v_t$, and the resulting field is then processed by the linear differential operator $\nabla \cdot$. 

The goal is to design a feedback controller $\K: (\rho_t,\mu) \mapsto v_t$ which renders $\rho = \mu$ a globally asymptotically stable (GAS) fixed point of the dynamics \eqref{continuity_eqn} for each possible setpoint $\mu$. In general, one may design this controller so that the velocity at each point $v_t(x)$ depends on the entire densities $\rho_t$ and $\mu$. This type of controller may be called \emph{centralized} or \emph{nonlocal}, since $v_t(x)$ may depend on data far away from $x$. (This sort of controller was investigated in our earlier work on this problem, for example \cite{emerick2023continuum,emerick2024causal,emerick2025optimal}.) In this paper, we are instead interested in \emph{distributed} or \emph{local} solutions to this problem. The setting we study considers controllers $\K$ which are spatial  differential operators, i.e., of the form
\begin{multline} \label{distributed_controller_eq}
    \K(\rho_t,\mu)(x) ~=~ k(x,\rho_t(x), \nabla \rho_t(x), \nabla^2 \rho_t(x), ... , \\
    \mu(x), \nabla \mu(x), \nabla^2 \mu(x), ...)
\end{multline}
for some (generally nonlinear) function $k$. This choice will be explained further in Section \ref{distributed_control_sec}.

We also point out that since the density $\rho$ is a continuum object (i.e.~it is infinite-dimensional), it matters which metric we choose to measure the stability in. We argue that the natural metric in our setting is the \emph{Wasserstein distance} from optimal transport theory. This will be discussed further in Section \ref{metrics_sec}.

Putting this all together, we thus pose the following basic \emph{distributed swarm regulation problem}.

\begin{prob}[Distributed Swarm Regulation] \label{distributed_regulation_prob}
    Determine a controller $\K$ of the form \eqref{distributed_controller_eq} which renders $\rho = \mu$ a globally asymptotically stable fixed point of the dynamics \eqref{continuity_eqn}, in the Wasserstein sense, for each possible setpoint $\mu$.
\end{prob}

\begin{defn}[Global Asymptotic Stability]
    A controller renders the setpoint $\mu$ \emph{globally asymptotically stable} with respect to a metric $d$ if it renders $\mu$ both: a) \emph{Lyapunov stable} with respect to $d$: for all $\varepsilon > 0$, there exists $\delta = \delta(\varepsilon) > 0$ such that for all initial conditions $\rho_0$ and resulting evolutions $\rho_t$,
    \begin{equation}
        d(\rho_0 , \mu) < \delta \qquad \Rightarrow \qquad d(\rho_t , \mu) < \varepsilon \quad \forall t ,
    \end{equation}
    and b) \emph{globally attractive} with respect to $d$: for all initial conditions $\rho_0$ and resulting evolutions $\rho_t$,
    \begin{equation}
        d(\rho_t,\mu) ~\to~ 0 \qquad \text{as} \qquad t ~\to~ \infty .
    \end{equation}
\end{defn}

We call the above problem ``basic'' since we often want to strengthen these requirements. For example, it is often desirable that $\K$ have \emph{finite control effort}, or that $\K$ render $\mu$ not just \emph{asymptotically} stable but \emph{uniformly exponentially} stable. We may also want to impose restrictions on the structure of $\K$, including on its symmetries, on its \emph{order} as a differential operator, and so on. Several such properties are listed in Section \ref{desirable_properties_sec}.

In Sections \ref{gradient_controller_sec} and \ref{main_results_sec}, we explore fundamental performance limitations of low-order distributed controllers. Specifically, we present several results showing that controllers of the form \eqref{distributed_controller_eq} which solve Problem \ref{distributed_regulation_prob} must have order $\geq 1$ to possess several natural symmetry, stability, and efficiency properties.

\subsection{Distributed Controllers as Differential Operators} \label{distributed_control_sec}

In the finite-agent setting, it is relatively clear how to formalize the notion of distributed control: the swarm of agents is endowed with a network structure, and each agent's control is allowed to depend on its own state as well as on the states of the other agents in its neighbor set. In the continuum setting, however, the correct formulation is not so obvious.

Drawing inspiration from the finite-agent case, one approach might be to define neighbor sets for each agent. A natural choice for the neighbor set for an agent at location $x$ might be the set of all agents in the ball of radius $\varepsilon$ centered at $x$, $B_{\varepsilon,x}$. In this formulation, the control $v_t(x)$ would be allowed to depend on the restrictions $\rho_t |_{B_{\varepsilon,x}}$ and $\mu |_{B_{\varepsilon,x}}$ (possibly in addition to $t$ and/or $x$ itself). However, while this is a very intuitive way to model this problem, it also comes with certain analytical challenges. Namely, the controllers are functions of infinite-dimensional inputs, and the resulting dynamics are governed by partial integro-differential equations or similar. This is because even for small values of $\varepsilon$, the controllers and resulting dynamics are nonlocal, since the control $v_t(x)$ may depend on data $\varepsilon$-far away from $x$.

For this reason, we take a different approach. The basic idea is that a fully local controller should depend only on information about $\rho_t$ and $\mu$ available in an arbitrarily small neighborhood of each point. In smooth settings, the practically usable information about a function $f$ in a small neighborhood of a point $x$ is captured by the set of derivatives, or the \emph{jet}, of $f$ at $x$.

\begin{defn}[Jet of a Function]
    The \emph{jet} of a $C^\infty$ function $f$ at a point $x$ is the set of all derivatives of $f$ at $x$:
    \begin{equation}
        (\J f)(x) ~:=~ \{ f(x), \nabla f(x), \nabla^2 f(x), ... \} .
    \end{equation}
    The truncation of this set to the first $m$ derivatives is called the \emph{$m$-jet} of $f$ at $x$:
    \begin{equation}
        (\J_m f)(x) ~:=~ \{ f(x) , \nabla f(x) , ... , \nabla^m f(x) \} .
    \end{equation}
\end{defn}

This motivates modeling distributed controllers as functions of the jets of $\rho$ and $\mu$, i.e., as functions of the form \eqref{distributed_controller_eq}, or equivalently, as (generally nonlinear) differential operators. In practice, we expect these controllers to only depend on a finite number of derivatives of $\rho$, $\mu$, i.e., to have finite \emph{order}.

\begin{defn}
    A controller $\K$ of the form \eqref{distributed_controller_eq} is said to be of \emph{order $m$} if it only depends only on the first $m$ derivatives of $\rho$, $\mu$. That is, if it has the form
    \begin{equation}
        \K (\rho_t,\mu)(x) ~=~ k \big(x,(\J_m \rho_t)(x), (\J_m \mu)(x) \big) .
    \end{equation}
\end{defn}

Modeling distributed controllers in this way has several advantages. First, the resulting controllers and dynamics are fully local. In particular, the dynamics are governed by differential (rather than integro-differential) equations. Second, controllers which are designed in the continuum setting may be translated back to the finite-agent setting using standard techniques (e.g., finite-difference stencils) to approximate the $m$-jets. Third, the order of the controller $\K$ provides a natural measure of the implementation overhead of $\K$. In particular, the order $m$ is a proxy for the \emph{degree of localization} in a finite-agent implementation, since higher order controllers require larger neighbor sets in order to approximate $\J_m f$. Additionally, higher-order controllers require smoother densities $\rho$ and $\mu$ and result in higher-order dynamics: if $\K$ is order $m$, then $\rho$ and $\mu$ must be at least $m$-times differentiable, and the resulting dynamics have order $m+1$.

This observation about the relationship between the order of $\K$ and the associated overhead implies that lower-order controllers are usually more desirable. This raises the following natural question: \emph{given some desired system behavior, what is the minimum order required for a distributed controller to achieve that behavior?} In the following sections, we attempt to answer that question for the problem of stabilization of densities around arbitrary setpoints.

\subsection{Distributed Control and Communication} \label{control_communication_sec}

Before moving on with this objective, however, we take a tangent to introduce a second, more general model of \emph{distributed control and communication}. To motivate this model, observe that the structure~\eqref{distributed_controller_eq} of $\K$ implies that the control is a local function of $\rho$ and $\mu$ only (possibly in addition to $t,x$). In other words, the controller is able to observe the nearby swarm density $\rho$ and environmental variable $\mu$, but (a) it \emph{does not communicate directly with other nearby agents/controllers}, and (b) {\em does not have its own temporal memory}. 

Thus, we may imagine a more general scenario where agents/controllers store some ``relevant information'' $\chi$ which may be communicated with neighboring agents and also used in determination of the control input. The variable $\chi$ may be viewed as a spatially distributed field of ``controller states'' $\chi(x)$. When agents update their own controller state through communication with their neighbors, the overall field $\chi$ evolves according to some dynamics
\begin{equation}    \label{K_state_eq}
    \partial_t \chi_t ~=~ \A(\rho_t,\mu,\chi_t)
\end{equation}
for some (possibly nonlinear) differential operator
\begin{equation}
    \A(\rho,\mu,\chi)(x) ~=~ a \big(x,(\J\rho)(x),(\J\mu)(x),(\J\chi)(x) \big) .
\end{equation}
The controller $\K$ may then be allowed to depend on local information about $\chi$ as well, taking the more general form
\begin{equation}    \label{K_out_eq}
    \K(\rho,\mu,\chi)(x) ~=~ k \big(x,(\J\rho)(x),(\J\mu)(x),(\J\chi)(x) \big) .
\end{equation}
Note that in this setting,~\eqref{K_state_eq} is the ``state equation'' of the controller, while~\eqref{K_out_eq} is the controller's ``output equation''. 

Thus, our model of \emph{distributed swarm control} consists of the single PDE
\begin{equation}
    \partial_t \rho_t ~=~ - \nabla \cdot \big( \rho_t \, \K(\rho_t,\mu) \big) , \label{control_only_eq}
\end{equation}
while our model of \emph{distributed swarm control and communication} consists of the coupled system
\begin{align}
    \partial_t \rho_t ~&=~ - \nabla \cdot \big(\rho_t \, \K(\rho_t,\mu,\chi_t) \big) \label{control_eq} \\
    \partial_t \chi_t ~&=~ \A(\rho_t,\mu,\chi_t) . \label{communication_eq}
\end{align}
Note that in the distributed swarm control model \eqref{control_only_eq}, the design parameter is the operator $\K$, while in the distributed swarm control and communication model \eqref{control_eq}-\eqref{communication_eq}, the design parameters are not only the operators $\K$ and $\A$, but also the \emph{structure of the controller state $\chi$}. The control architecture~\eqref{control_eq}-\eqref{communication_eq} is depicted in Figure~\ref{rho_chi.fig}. 

\begin{figure}[t]
    \centering
    \includegraphics[width=0.45\textwidth]{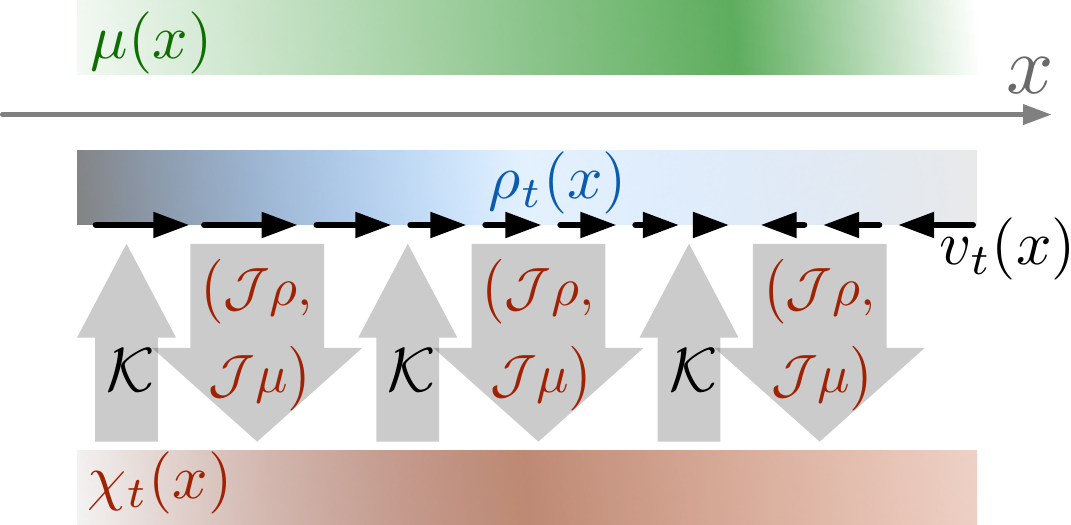} 

    \caption{Depiction of the distributed control architecture~\eqref{control_eq}-\eqref{communication_eq}. The density $\rho_t$ is transported by the ``control'' velocity vector field $v_t$ in order to stabilize around a reference density $\mu$. $\chi_t$ is the spatially-distributed controller state which evolves according to its own PDE whose inputs are the jets $({\cal J}\rho,{\cal J}\mu)$. The controller $\K$ determines the velocity field $v_t$. All mappings between system and controller depicted by gray wide arrows are actually continuum and local. 
        } 
  \label{rho_chi.fig} 
\end{figure} 

The distinction between the above models is analogous to the distinction between \emph{static} and \emph{dynamic} controllers in the classical sense. The controller in \eqref{control_only_eq} is ``memoryless'' in that it is a static function of the state $\rho$ (and setpoint $\mu$), while the controller in \eqref{control_eq}-\eqref{communication_eq} is an external dynamical system that we couple to the system we wish to control.

Many distributed control algorithms proposed in the literature can be represented using this more general model. In particular, a popular choice is to take each $\chi(x)$ to be an approximation of the global state and relevant environmental information $(\rho,\mu)$. That is, $\chi(x) = (\tilde{\rho}_x,\tilde{\mu}_x)$ is the agent at location $x$'s approximation of the entire (i.e.~global) densities $\rho$ and $\mu$. Thus, $\chi$ is a map from the spatial domain $\Omega$ to two copies of the space $\mathcal{P}(\Omega)$ of probability densities on $\Omega$. The communication dynamics \eqref{communication_eq} are designed so that each $(\tilde{\rho}_x,\tilde{\mu}_x)$ converges to the true global quantities $(\rho,\mu)$ (so long as $\rho$ is not changing). Then, the controller $\K$ may be chosen to be the optimal \emph{centralized} control law $\bar{K}$, only applied to the estimates $(\tilde{\rho}_x,\tilde{\mu}_x)$ instead of the true global quantities $(\rho,\mu)$:
\begin{equation}
    \K(\rho,\mu,\chi)(x) ~=~ \bar{K} (\chi(x)) ~=~ \bar{K} (\tilde{\rho}_x,\tilde{\mu}_x) .
\end{equation}
This approach may be considered to be a distributed analogue of \emph{certainty-equivalent control}: each agent receives local (partial) measurements and constructs an estimate of the true global state through the dynamics \eqref{communication_eq}. This estimate is then used in place of the true state in the optimal full-state feedback control $\bar{K}$.

In the rest of the paper, we will focus only on the simpler distributed swarm control model \eqref{control_only_eq} (with $\K$ taking the form \eqref{distributed_controller_eq}). We will leave the study of the more general distributed swarm control and communication model \eqref{control_eq}-\eqref{communication_eq} to future work.

\subsection{Metrics and Stability} \label{metrics_sec}

As mentioned earlier, since the density $\rho$ is a continuum (i.e.~infinite-dimensional) object, it matters which metric we choose to measure the stability in. The most well-known and widely used metric for measuring the distance between functions $\rho$, $\mu$ is of course the $L^p$ distance:
\begin{equation}
    \| \rho - \mu \|_{L^P}^p ~:=~ \int_\Omega | \rho(x) - \mu(x) |^p \, dx .
\end{equation}
However, in the setting of robotic swarms, we argue that this is not the natural metric to use. This is because the $L^p$ distance is essentially a ``pointwise'' metric: it picks up the differences between $\rho$ and $\mu$ at each point, but it is ``blind'' to the relationships between those points, i.e., to the structure of the underlying spatial domain $\Omega$.

However, in the case of physical swarms, we care very much about the structure of the spatial domain. The domain $\Omega$ comes equipped with a distance $d_\Omega$, and we consider an agent at location $x(t)$ to be converging to a target location $y$ if $d_\Omega(x(t),y) \to 0$. The fact that the $L^p$ distance is a pointwise metric means that it is insensitive to this sort of convergence\footnote{Indeed, given any target density $\mu$, it is possible to find a perturbation $\tilde{\mu}$ of the form $\tilde{\mu}(x) := \mu(\phi(x))$ for which $\sup_x d_\Omega(x,\phi(x)) \leq \varepsilon$, but $d_{L^p}(\mu,\tilde{\mu}) = O(1)$.}.

Thus, we seek a different metric in which to measure this convergence/stability, and the \emph{Wasserstein distance} from optimal transport theory\footnote{For a quick introduction, the authors recommend \cite{santambrogio2014introduction}.} captures our intuition much better here. It may be defined\footnote{Note that as it is written, $\gamma$ may be singular. In other words, \eqref{wass_dist_def} and \eqref{transport_plan_eq} should be interpreted in the sense of distributions. This definition is equivalent to the usual measure-theoretic one.} as
\begin{equation} \label{wass_dist_def}
W_p^p(\rho,\mu) ~:=~ \inf_{\gamma\in\Gamma(\rho,\mu)}
\int_{\Omega\times\Omega}\|x-y\|_p^p \, \gamma(x,y) \, dx \, dy ,
\end{equation}
where $\Gamma(\rho,\mu)$ denotes the set of \emph{transport plans} (i.e., joint distributions) between the densities $\rho$ and $\mu$:
\begin{equation} \label{transport_plan_eq}
    \Gamma(\rho,\mu) ~:=~ \bigl\{\gamma ~:~ \gamma \geq 0, ~ \textstyle\int \gamma \, dy =\rho, ~ \int \gamma \, dx = \mu \bigr\} .
\end{equation}
Note that $\rho$ and $\mu$ must have the same mass in order for this quantity to be defined, and without loss of generality, we may assume this mass to be 1. In other words, we assume that $\rho$ and $\mu$ are \emph{probability densities}.

There is a sense in which the Wasserstein distance is the canonical way in which to measure this sort of convergence, since it \emph{extends} the distance function $\|x-y\|_p$ from $\Omega$ to the space $\mathcal{P}(\Omega)$ of probability densities on $\Omega$. That is, the Wasserstein distance between Dirac masses reduces to the distance between their base points: $W_p(\delta_x,\delta_y) = \| x - y \|_p$.

In this paper, we will assume that the swarm $\rho$ and target $\mu$ live within a compact convex subset $\Omega \subset \R^n$. We will endow this subset with the standard Euclidean distance $\| x-y \|_2$, and thus we will endow the space $\mathcal{P}(\Omega)$ with the $2$-Wasserstein distance (that is, the Wasserstein distance as defined above for the specific choice $p=2$).

It is worth pointing out that on compact domains, the $L^p$ distance is a \emph{stronger} notion of distance than the Wasserstein distance. In particular, we have the inequality
\begin{equation}
    W_2^2(\rho,\mu) ~\leq~ \frac{1}{2} \, \operatorname{diam}(\Omega)^2 \, \operatorname{vol}(\Omega)^{1-\frac{1}{p}} \, \|\rho-\mu\|_{L^p} ,
\end{equation}
where $\operatorname{diam}$ denotes the \emph{diameter} of $\Omega$ and $\operatorname{vol}$ its volume. (This inequality is derived in Appendix \ref{technical_preliminaries_appendix}.) Thus, if $\rho$ is converging to $\mu$ in $L^p$, then it is also converging in $W_2$. However, the reverse does not necessarily hold. In other words, the Wasserstein distance measures a sort of \emph{weak} convergence.

\subsection{Other Desirable Properties of Controllers} \label{desirable_properties_sec}

Here, we list some other properties that we may want a controller to possess, beyond the simple \emph{global asymptotic stability (GAS)} or \emph{low order} already discussed. These include both stronger stability properties and structural/implementation-related properties.

\begin{defn}[Uniform Exponential Stability]
    A controller renders the setpoint $\mu$ \emph{uniformly locally exponentially stable (ULES)} with respect to a metric $d$ if there exist constants $\varepsilon > 0$, $C \geq 1$, and $\lambda > 0$ such that for all initial conditions $\rho_0$ with $d(\rho_0,\mu) < \varepsilon$, the evolution of $\rho_t$ satisfies
    \begin{equation} \label{ues_eq}
        d(\rho_t,\mu) ~\leq~ C e^{-\lambda t} \, d(\rho_0,\mu) .
    \end{equation}
    A controller renders the setpoint \emph{uniformly globally exponentially stable (UGES)} with respect to $d$ if \eqref{ues_eq} holds for all initial conditions $\rho_0$, independent of any $\varepsilon$.
\end{defn}

This is the type of strong stability that exists in stable finite-dimensional linear systems. Additionally, this is the type of stability in infinite-dimensional and nonlinear systems which may be detected through linearization.

\begin{defn}[Finite Control Effort]
    A controller $\K$ has \emph{finite control effort} with respect to a norm $\| \cdot \|$ if for every initial condition $\rho_0$, the evolution of $\rho_t$ is such that
    \begin{equation}
        \int_0^\infty \| \K(\rho_t,\mu ) \| \, dt ~<~ \infty .
    \end{equation}
\end{defn}

In particular, if a sufficiently regular controller $\K$ renders the setpoint $\mu$ GAS with finite control effort, then $\K(\rho,\mu) = 0$ at the setpoint $\rho = \mu$, and $\K(\rho_t,\mu) \to 0$ as $\rho_t \to \mu$ along trajectories.

Some examples of norms which we may be interested in include the $L^2$ norm $\| \cdot \|_{L^2}$ or kinetic energy $\| \cdot \|_{L^2(\rho)}$.

\begin{defn}[Symmetries/Equivariance] \label{symmetry_def}
    A controller $\K$ is \emph{symmetric/equivariant with respect to a group of transformations $G$} if $g_* \K = \K g_*$ for all $g \in G$.
\end{defn}

Here, we suppose that $G$ is a group of transformations from $\R^n \to \R^n$, such as a spatial translations or rotations. The symbol $g_*$ denotes the (tensorial\footnote{While the pushforward action may be defined for any tensor-valued quantity, we do not need the full generality here, since we only consider controllers of order at most 1 in this paper, and thus never deal with objects more complicated than vectors.}) pushforward by the map $g$, which may act on different types of objects in different ways. For example, if $f$ is a scalar function, then $g_*$ acts by $(g_*f)(x) =f(g^{-1}(x))$, whereas if $v$ is a vector field, then $g$ acts by $(g_*v)(x) = Dg(g^{-1}(x))[v(g^{-1}(x))]$, where $DF(y)[u]$ denotes the differential of the function $F$, evaluated at the point $y$, acting on the vector $u$.

In Appendix \ref{symmetries_appendix}, we show that a controller $\K$ of the form \eqref{distributed_controller_eq} is symmetric with respect to spatial translations if and only if it does not depend explicitly on $x$. (Similarly, controllers of the form \eqref{distributed_controller_eq} are always symmetric with respect to time translations because they do not depend explicitly on $t$.) Controllers which are symmetric with respect to spatial rotations have a more complicated form, and are treated in the appendix.

The practical motivation for considering these symmetries is that if the controller depends explicitly on any of these parameters, then it requires precise calibration together with some externally-provided references such as central clocks or GPS systems. There are many applications where these sorts of measurements cannot be expected to be available. It may also be argued that these controllers are not truly distributed since they depend on some external/centralized information.

\section{Example: Error-Gradient Controller} \label{gradient_controller_sec}

Here, we present a simple, order-1, distributed control law\footnote{Note that this controller was investigated previously in \cite{eren2017velocity}. It appears that the present treatment fixes a gap in that work regarding the positivity of the density $\rho$ (see Appendix \ref{error_gradient_appendix} for details). Thank you to Karthik Elamvazhuthi for pointing this out to us.} which solves Problem \ref{distributed_regulation_prob} with all the desirable properties mentioned in the previous section.

\begin{prop} \label{gradient_controller_prop}
    The following \emph{error-gradient} controller
    \begin{equation} \label{gradient_controller_eq}
        \K(\rho,\mu) ~=~ - \frac{1}{\rho} \nabla (\rho - \mu)
    \end{equation}
    solves Problem \ref{distributed_regulation_prob} for sufficiently nice $\rho_0$, $\mu$.
\end{prop}

\begin{sketch}
Define the error by $e := \rho - \mu$ and consider the error dynamics:
\begin{multline}
    \partial_t e ~=~ \partial_t (\rho - \mu) ~=~ \partial_t \rho - 0 ~=~ - \nabla \cdot (\rho \, \K(\rho,\mu)) \\
    ~=~ \nabla \cdot \left( \rho \, \frac{1}{\rho} \, \nabla (\rho - \mu) \right) ~=~ \nabla \cdot \nabla e ~=~ \Delta e .
\end{multline}
In other words, the error dynamics are exactly the heat equation. Since $\rho$ and $\mu$ have mass 1, $e$ has mean 0. Thus every $e$ converges to 0, rendering $\mu$ GAS. See Appendix \ref{error_gradient_appendix} for a rigorous treatment.
\end{sketch}

Some remarks are in order here. First, we call the controller \eqref{gradient_controller_eq} ``error-gradient'' because it acts so that agents descend the local gradient of the error $\nabla e = \nabla(\rho-\mu)$. The weighting by $1/\rho$ may be interpreted as a form of \emph{feedback linearization}, since it cancels the factor of $\rho$ in the continuity equation to yield linear closed-loop dynamics. This is done mainly to simplify the resulting analysis. Second, this is a desirable controller in almost every respect: it is UGES both in the $W_2$ and the stronger $L^p$ senses, it has finite control effort, and it has all of the symmetries that we would hope for. It also has order 1, depending only on a single derivative of $\rho$, $\mu$.

\begin{prop} \label{gradient_controller_uges_prop}
    The controller \eqref{gradient_controller_eq} renders the setpoint $\mu$ uniformly globally exponentially stable in both the $W_2$ and $L^p$ senses.
\end{prop}

\begin{prop} \label{gradient_controller_effort_prop}
    The controller \eqref{gradient_controller_eq} has finite control effort with respect to the $L^2$ norm.
\end{prop}

\begin{prop} \label{gradient_controller_symmetries_prop}
    The controller \eqref{gradient_controller_eq} is symmetric with respect to spatial translations and spatial rotations.
\end{prop}

Proofs of these propositions are provided in Appendix \ref{error_gradient_appendix}.

\section{Fundamental Limitations of Pointwise Controllers} \label{main_results_sec}

Following the example of the previous section, the natural question is then: \emph{can we do better?} That is, does there exist a controller with order 0 which achieves these same behaviors? As may be expected simply from intuition, the answer is ``no''. However, the story is nuanced, and exposing this nuance sheds light on the nature of the problem more generally. The main idea is stated in the following ``heuristic''.

\begin{heur}
    Pointwise (i.e., order-0) controllers may solve Problem \ref{distributed_regulation_prob}, but only through mechanisms which are based on \emph{mixing}. These mechanisms are much weaker and less structured than those available to controllers with order $\geq 1$. In particular, pointwise controllers may only stabilize in a weak sense, and they do not possess finite control effort or the natural symmetries listed in Section \ref{desirable_properties_sec}.
\end{heur}

The results of this section are organized into three parts, addressing the symmetries of such controllers, their stability properties, and the mixing mechanisms that they must ultimately be based on.

We begin by expanding on our intuition about why it seems that pointwise controllers should not possess the same sorts of desirable properties as controllers with order $\geq 1$. Observe that if a controller $\K$ depends on any derivatives of $\rho$, $\mu$ -- as in \eqref{gradient_controller_eq} -- then the control $v_t(x)$ depends on at least some information about $\rho$, $\mu$ in a small neighborhood of $x$. (Equivalently, in a finite-agent approximation, each agent has a nonempty neighbor set.) However, any pointwise/order-0 controller
\begin{equation} \label{pointwise_control_eq}
    \K(\rho,\mu)(x) ~=~ k(x,\rho(x),\mu(x))
\end{equation}
is ``fully decentralized'' in the sense that the control $v_t(x)$ depends only on the information exactly at $x$ (not in a small neighborhood). (Equivalently, in a finite-agent approximation, each neighbor set is empty.) Intuitively, this blocks control mechanisms which leverage the local geometry of $\rho$, $\mu$ or information exchange between controllers.

In the following sections, we attempt to formalize these intuitions and investigate their consequences, ultimately arriving at the conclusion that the only stabilization mechanism available to pointwise controllers is \emph{mixing}.

\subsection{Symmetries and Linearized Dynamics}

First, we investigate the symmetries that pointwise controllers can possess. For example, it is immediate that this pointwise structure is incompatible with rotational symmetry.

\begin{prop} \label{rotational_controller_prop}
    The only controller of the form \eqref{pointwise_control_eq} which is symmetric under spatial rotations is the trivial controller
    \begin{equation}
        k(x,\rho(x),\mu(x)) ~=~ 0 .
    \end{equation}
\end{prop}

\begin{sketch}
Consider a rotation about $x$. Since $x$, $\rho(x)$, and $\mu(x)$ are all fixed by this rotation, so must be the vector $k(x,\rho(x),\mu(x))$. However, the only vector which is fixed by all rotations is 0. See Appendix \ref{symmetries_appendix} for a rigorous treatment.
\end{sketch}

\begin{cor} \label{rotational_controller_cor}
    No controller of the form \eqref{pointwise_control_eq} which is symmetric under spatial rotations solves Problem \ref{distributed_regulation_prob}.
\end{cor}

Note that the above proposition does not apply to controllers with order $\geq 1$, since it is only the pointwise nature of \eqref{pointwise_control_eq} which is in conflict. In other words, dependence of $\K$ on derivatives of $\rho$, $\mu$ allows for nontrivial rotationally-symmetric controllers (e.g., \eqref{gradient_controller_eq}).

A similar result also holds for controllers that are symmetric under spatial translations, but takes a bit more work. We first consider the \emph{linearized error dynamics}.

\begin{lem}[Linearized Error Dynamics] \label{linearized_dynamics_lem}
    Consider the dynamics \eqref{continuity_eqn} evolving under feedback with the control \eqref{pointwise_control_eq}. Suppose that $\rho = \mu$ is a fixed point of the closed loop, and define the \emph{error} by $e_t := \rho_t - \mu$. Then the dynamics of the error, linearized around the point $e=0$, are given by
    \begin{equation} \label{linearized_dynamics_eq}
        \partial_t e_t ~=~ -\nabla \cdot (e_t \, b)
    \end{equation}
    for a vector field
    \begin{equation} \label{vector_field_b_eq}
        b(x) ~:=~ k(x,\mu(x),\mu(x)) + \mu(x) \, \partial_\rho k(x,\mu(x),\mu(x)) ,
    \end{equation}
    where $\partial_\rho k$ denotes the partial derivative of $k$ with respect to the second argument.
\end{lem}

\begin{sketch}
    Rewrite the dynamics in terms of the error
    \begin{equation} \label{full_error_dynamics_eq}
        \partial_t e_t ~=~ - \nabla \cdot \big( (e_t + \mu) \, \K(e_t+\mu ,\mu) \big)
    \end{equation}
    and expand $\K$ around the setpoint $\mu$
    \begin{equation}
        \K(e_t+\mu,\mu) ~=~ \K(\mu,\mu) + e_t \partial_\rho \K(\mu,\mu) + o(\|e_t\|) .
    \end{equation}
    Using the fact that $\mu$ is a fixed point and taking $\| e \|$ to be small, we obtain the linear approximation \eqref{linearized_dynamics_eq}-\eqref{vector_field_b_eq}. See Appendix \ref{linearized_dynamics_appendix} for a rigorous treatment.
\end{sketch}

The linearized error dynamics take exactly the form of the continuity equation, describing the transport of $e$ by the vector field $b$. Importantly, $b$ is independent of $e$.

\begin{prop} \label{translational_controller_prop}
    No controller of the form \eqref{pointwise_control_eq} which is symmetric under spatial translations solves Problem \ref{distributed_regulation_prob} with $W_2$ uniform local exponential stability.
\end{prop}

\begin{sketch}
We may construct a counterexample as follows. Suppose that $\K$ is symmetric under spatial translations so that it does not depend explicitly on $x$, and take $\mu$ to be constant. Then $b$ is also constant. By the zero-flux boundary condition\footnote{Under more general boundary conditions, $e$ may evolve according to pure translation: $e_t(x) = e_0(x-tb)$, but the conclusion is the same.}, $b$ must be zero. Then $e$ must be constant, the linearization is unstable, and the nonlinear dynamics cannot be ULES. See Appendix \ref{symmetries_appendix} for a rigorous treatment.
\end{sketch}

Thus, for a controller of the form \eqref{pointwise_control_eq} to solve Problem \ref{distributed_regulation_prob}, it must depend explicitly on $x$.

\subsection{Conserved Quantities}

Taking a closer look at the linearized error dynamics \eqref{linearized_dynamics_eq}, we may observe a fundamental tension: in order for a controller to stabilize the fixed point $\rho = \mu$, it needs to be able to dissipate errors. However, the continuity equation is a fundamentally conservative, transport-type equation.

\begin{prop} \label{l1_error_prop}
    The linearized error dynamics \eqref{linearized_dynamics_eq} conserve the $L^1$ norm of the error: $\| e_t \|_{L^1} =$ constant.
\end{prop}

\begin{cor} \label{lp_instability_cor}
    No controller of the form \eqref{pointwise_control_eq} can solve Problem \ref{distributed_regulation_prob} with $L^p$ uniform local exponential stability.
\end{cor}

See Appendix \ref{linearized_dynamics_appendix} for proofs of these results. Note that while this does not show that a pointwise controller cannot solve Problem \ref{distributed_regulation_prob}, it does show that such stabilization can occur only in a weak (i.e.~Wasserstein) sense.

\subsection{Mixing Dynamics}

The results of the previous sections show us that pointwise controllers can only ever transport error -- never dissipate it -- and thus the error can never approach zero in any strong/pointwise sense. However, the error may approach zero in a weak sense if regions of positive and negative error can be brought close together in a small-scale, high-frequency manner. This suggests \emph{mixing} as a potential mechanism by which pointwise controllers may solve Problem \ref{distributed_regulation_prob}, and indeed, we have the following results.

\begin{thm} \label{mixing_thm}
Zero is a $W_2$ globally attractive fixed point of the linearized dynamics
\begin{equation}
    \partial_t e_t ~=~ -\nabla \cdot (e_t \, b)
\end{equation}
if and only if the vector field $b$ is mixing with respect to a sufficiently regular invariant measure.
\end{thm}

\begin{defn}[Mixing]
    A vector field $b$ on a compact domain $\Omega$ is said to be \emph{mixing} if there exists an invariant measure $\pi$ on $\Omega$: $(\phi_t)_\# \pi = \pi$, such that for any two continuous functions $f,g$ on $\Omega$,
    \begin{equation} \label{mixing_eq}
        \lb f , g \circ \phi_t^{-1} \rb_{L^2(\pi)} ~\to~ \lb 1 , f \rb_{L^2(\pi)} \, \lb 1 , g \rb_{L^2(\pi)} ,
    \end{equation}
    where $\phi$ denotes the \emph{flow} of $b$:
    \begin{equation} \label{flow_map_eq}
        \partial_t \phi_t(x) ~=~ b(\phi_t(x)), \quad \phi_0(x) = x ,
    \end{equation}
    and $(\cdot)_\#$ the \emph{measure pushforward}:
    \begin{equation} \label{pushforward_eq}
        \int_A (\phi_t)_\# \pi  ~=~ \int_{\phi_t^{-1}(A)} \pi  \qquad \forall A \subset \Omega ~~ \text{meas} .
    \end{equation}
\end{defn}

\begin{sketch}
    The convergence of $\rho_t \to \mu$ in the Wasserstein sense (equivalently, $e_t \to 0$) is equivalent to weak convergence in duality with continuous functions: $\lb \psi , \rho_t \rb \to \lb \psi , \mu \rb$ (equivalently, $\lb \psi , e_t \rb \to 0$) for all continuous functions $\psi$. Identifying $\psi$ with $f$ and $e$ with $g$ (and accounting for the invariant measure $\pi$), the mixing property \eqref{mixing_eq} implies that $e=0$ is attractive since $e$ is zero mean: $\lb 1 , e_t \rb = 0$. The converse follows by breaking $g$ into its average and zero-mean components. See Appendix \ref{mixing_appendix} for a rigorous treatment.
\end{sketch}

\begin{cor} \label{exponential_mixing_cor}
    For a controller of the form \eqref{pointwise_control_eq} to solve Problem \ref{distributed_regulation_prob} with $W_2$ uniform local exponential stability, the vector field \eqref{vector_field_b_eq} must be mixing with respect to a sufficiently regular invariant measure.
\end{cor}

See Appendix \ref{mixing_appendix} for proofs of these results. We emphasize that this limits the existence of pointwise solutions to Problem \ref{distributed_regulation_prob} in low dimensions in particular. There are no mixing vector fields in one dimension, and there are strong topological constraints on mixing vector fields in two dimensions (due to the Poincare-Bendixon Theorem, for example). In dimensions three and higher, mixing vector fields exist in abundance, but can be difficult to construct.

Lastly, we emphasize that the results stated in this section apply only to \emph{deterministic} pointwise controllers. In particular, \emph{random} pointwise controllers may solve Problem \ref{distributed_regulation_prob} by introducing a Laplacian term into the dynamics so that errors may be dissipated (see \cite{bandyopadhyay2017probabilistic}, for instance). The study of these sorts of controllers is left to future work.

\section{Conclusion} \label{conclusion_sec}

In this paper, we formulated a problem of distributed regulation for continuum swarms and proposed a model of distributed controllers as differential operators. Using this model, we investigated fundamental performance limitations of low-order distributed controllers, showing that static, pointwise controllers which solve the proposed problem must be based on mixing. Such controllers either do not exist (in sufficiently low dimensions) or are difficult to construct. Furthermore, even when such controllers do exist, they do not possess desirable symmetry or strong stability properties. This motivates the investigation of more general distributed control architectures which involve either a) higher controller order, or b) dynamic feedback (as outlined in Section \ref{control_communication_sec}).


\appendix

\makeatletter
\@addtoreset{thm}{subsection}
\makeatother
\renewcommand{\thethm}{\thesubsection\arabic{thm}}

In this appendix, we provide rigorous statements and proofs for all results. Restatements of results appearing in the main text appear with the same number and in parentheses, e.g., (\textbf{Proposition 8}) for the restatement of Proposition 8. Results which appear only in the appendix are numbered sequentially starting with the appendix letter and without parentheses, e.g., \textbf{Lemma A1} for the first technical lemma appearing in Appendix A.

\subsection{Modes of Convergence} \label{technical_preliminaries_appendix}

Here, we list some results which characterize the relationship between convergence in $W_2$ and other modes of convergence, including convergence in $L^p$, $H^{-1}$ (Sobolev order minus-1), and weak convergence (i.e., in duality with continuous functions).

\begin{lem}
    Let $\Omega \subset \R^d$ be bounded and measurable and $\rho$, $\mu$ be absolutely continuous probability measures on $\Omega$. (By abuse of notation, denote also by $\rho$, $\mu$ their density functions.) Then it holds that
    \begin{equation}
        W_2^2(\rho,\mu) ~\leq~ \frac{1}{2} \, \operatorname{diam}(\Omega)^2 \, \operatorname{vol}(\Omega)^{1-\frac{1}{p}} \, \|\rho-\mu\|_{L^p} ,
    \end{equation}
    where $\operatorname{diam}$ denotes the \emph{diameter} of $\Omega$ and $\operatorname{vol}$ its volume.
\end{lem}

\begin{proof}
    Supposing that $\Omega$ is bounded, we have
    \begin{equation}
        W_2^2(\rho,\mu) ~\leq~ \operatorname{diam}(\Omega)^2 \, \| \rho - \mu \|_\text{TV} ,
    \end{equation}
    where $\| \cdot \|_\text{TV}$ denotes the total variation norm. Supposing further that $\rho$ and $\mu$ are absolutely continuous, the Holder inequality gives
    \begin{equation}
        \| \rho - \mu \|_{TV} ~=~ \frac{1}{2} \, \| \rho - \mu \|_{L^1} ~\leq~ \frac{1}{2} \, \operatorname{vol}(\Omega)^{1-\frac{1}{p}} \, \| \rho - \mu \|_{L^p} . 
    \end{equation}
    Putting these together, we obtain the result.
\end{proof}

\begin{lem}[{\cite[Theorem 5.10]{Santambrogio2015}}] \label{weak_convergence_lem}
    Let $\Omega \subset \R^d$ be compact and $p \in [1,\infty)$. Then $W_p(\mu_n,\mu) \to 0$ if and only if $\mu_n \to \mu$ weakly, that is, $\lb f , \mu_n \rb \to \lb f, \mu \rb$ for all $f \in C(\Omega)$.
\end{lem}

\begin{lem}[{\cite[Theorem 5.34]{Santambrogio2015}}] \label{h_minus_one_lem}
    Suppose that $\mu$ and $\nu$ are absolutely continuous measures on a convex domain $\Omega$ with both densities bounded below and above by constants $a,b$ with $0 < a < b < \infty$. Then
    \begin{equation}
        b^{1/2} \, \| \mu -  \nu \|_{H^{-1}} ~\leq W_2(\mu,\nu) ~\leq~ a^{-1/2} \, \| \mu - \nu \|_{H^{-1}} .
    \end{equation}
\end{lem}

\subsection{Analysis of Error-Gradient Controller} \label{error_gradient_appendix}

In order to formally establish the asymptotic convergence result for the controller proposed in Proposition \ref{gradient_controller_prop}, we present a rigorous version of the proposition, which first addresses some basic well-posedness results of the solution of the resulting PDE.

We consider only classical solutions here, and thus suppose that $\rho_0$, $\mu$ are sufficiently regular so that all derivatives are defined and the evolution $\rho_t$ is uniformly positive\footnote{We point out that while the resulting dynamics are always well-posed, solutions may not preserve the positivity of $\rho$ if $\mu$ has large isolated spikes. The assumption \eqref{positivity_assumption} rules out this behavior but is indeed restrictive, and one goal of future work is to relax this assumption.}.

\begin{assumption}
\label{first_assumptions}
    The following results in Appendix \ref{error_gradient_appendix} are proven under the following assumptions:
    \begin{enumerate}
        \item $\Omega \subset \R^n$ is compact and convex with smooth boundary
        \item $v$ satisfies the no-flux condition $v \cdot \hat{n} = 0$ on $\partial \Omega$
        \item $\rho_0$ and $\mu$ are $C^2$ and satisfy the inequality
        \begin{equation} \label{positivity_assumption}
            \min_\Omega (\mu) + \min_\Omega (\rho_0 - \mu) ~\geq~ c_0 ~>~ 0 .
        \end{equation}
    \end{enumerate}
\end{assumption}

\noindent(\textbf{Proposition \ref{gradient_controller_prop}}.)
Assume the hypotheses in Assumption \ref{first_assumptions}, and consider the feedback law
\begin{equation}\label{eq:feedback_heat}
  k(\rho,\mu)(t,x)
  :=
  -\frac{1}{\rho(t,x)}\,\nabla\bigl(\rho(t,x)-\mu(x)\bigr).
\end{equation}
Let $\rho$ be a classical solution of the continuity equation
\begin{equation}\label{eq:closed_loop_continuity}
\begin{cases}
\partial_t \rho = - \nabla \cdot (\rho\, k(\rho,\mu))
& \text{in } (0,T)\times\Omega,\\[0.3em]
(\rho\,k(\rho,\mu))\cdot n = 0
& \text{on } (0,T)\times\partial\Omega,\\[0.3em]
\rho(0,\cdot)=\rho_0
& \text{in }\Omega.
\end{cases}
\end{equation}
Then:
\begin{enumerate}
    \item $\rho(t,x)\ge c_0$ for all $(t,x)\in[0,T]\times\Omega$.
    \item $W_2(\rho_t,\mu) \le \sqrt{\frac{b}{c_0}}\,e^{-\lambda_1 t}W_2(\rho_0,\mu)$, where $b:=\|\rho_0\|_{L^\infty}$  i.e., the controller in \eqref{eq:feedback_heat} renders $\mu$ globally asymptotically stable. 
\end{enumerate}

\begin{proof}

\noindent(1) \textbf{Nonnegativity condition.} Since
\[
\rho\,k(\rho,\mu)=-\nabla(\rho-\mu),
\]
the continuity equation rewrites as
\[
\partial_t\rho-\Delta\rho=-\Delta\mu
\qquad\text{in }(0,T)\times\Omega.
\]
Moreover, the no-flux boundary condition yields
\begin{equation}\label{eq:boundary_condition}
0=(\rho \,\K (\rho,\mu))\cdot n=-\partial_n(\rho-\mu),
\end{equation}
Hence $\rho$ solves the Neumann problem
\[
\begin{cases}
\partial_t\rho-\Delta\rho=-\Delta\mu & \text{in }(0,T)\times\Omega,\\
\partial_n(\rho-\mu)=0 & \text{on }(0,T)\times\partial\Omega,\\
\rho(0,\cdot)=\rho_0 & \text{in }\Omega.
\end{cases}
\]
Now define the error $e:=\rho-\mu$. Since $\mu$ is time-independent, we have
\[
\partial_t e=\partial_t\rho=\Delta(\rho-\mu)=\Delta e
\qquad\text{in }(0,T)\times\Omega,
\]
and, because $\partial_n\rho=\partial_n\mu=0$ on $\partial\Omega$,
\[
\partial_n e=0
\qquad\text{on }(0,T)\times\partial\Omega.
\]
Thus $e$ solves the homogeneous Neumann heat equation
\[
\begin{cases}
\partial_t e-\Delta e=0 & \text{in }(0,T)\times\Omega,\\
\partial_n e=0 & \text{on }(0,T)\times\partial\Omega,\\
e(0,\cdot)=\rho_0-\mu & \text{in }\Omega.
\end{cases}
\]

By the maximum principle,
\[
e(t,x)\ge \min_{\Omega}(\rho_0-\mu)
\ge \min_{\Omega}\rho_0-\max_{\Omega}\mu.
\]
Therefore,
\[
\rho(t,x)=\mu(x)+e(t,x)
\ge \min_\Omega \mu(x)+\min_{\Omega}(\rho_0 -\mu).
\]
Using Assumption \ref{first_assumptions}, we have $\rho(t,x) \ge c_0$.

\noindent(2) \textbf{UGES property.}
Since \(\rho\) and \(\mu\) have the same mass,
\[
\int_\Omega e_t\,dx=0
\qquad \text{for all }t\ge 0.
\]
Let \(\phi_t\) be the mean-zero solution of
\[
-\Delta \phi_t=e_t.
\]
Then
\[
\|e_t\|_{\dot H^{-1}(\Omega)}^2
=\int_\Omega |\nabla \phi_t|^2\,dx
=\int_\Omega e_t\phi_t\,dx.
\]
Differentiating in time and using \(\partial_t e_t=\Delta e_t\), we get
\[
\frac12\frac{d}{dt}\|e_t\|_{\dot H^{-1}(\Omega)}^2
=\int_\Omega (\partial_t e_t)\phi_t\,dx
=\int_\Omega (\Delta e_t)\phi_t\,dx.
\]
Integrating by parts and using \(e_t=-\Delta\phi_t\),
\[
\int_\Omega (\Delta e_t)\phi_t\,dx
=-\int_\Omega \nabla e_t\cdot \nabla \phi_t\,dx
=-\int_\Omega e_t^2\,dx.
\]
Hence
\[
\frac12\frac{d}{dt}\|e_t\|_{\dot H^{-1}(\Omega)}^2
=-\|e_t\|_{L^2(\Omega)}^2.
\]
Since \(e_t\) has zero mean, Poincar\'e's inequality yields
\[
\|e_t\|_{L^2(\Omega)}^2
\ge \lambda_1 \|e_t\|_{\dot H^{-1}(\Omega)}^2,
\]
where \(\lambda_1>0\) denotes the first nonzero eigenvalue of \(-\Delta\) on \(\Omega\)
with Neumann boundary conditions. Therefore
\[
\frac{d}{dt}\|e_t\|_{\dot H^{-1}(\Omega)}^2
\le -2\lambda_1 \|e_t\|_{\dot H^{-1}(\Omega)}^2,
\]
and Gr\"onwall's lemma gives
\[
\|e_t\|_{\dot H^{-1}(\Omega)}
\le e^{-\lambda_1 t}\|e_0\|_{\dot H^{-1}(\Omega)}.
\]
Equivalently,
\[
\|\rho_t-\mu\|_{\dot H^{-1}(\Omega)}
\le e^{-\lambda_1 t}\|\rho_0-\mu\|_{\dot H^{-1}(\Omega)}.
\]

By the maximum principle,
\[
\min_{\Omega}(\rho_0-\mu)\le e(t,x)\le \max_{\Omega}(\rho_0-\mu)
\;\forall (t,x)\in[0,T]\times\Omega.
\]
Therefore,
\[
\mu(x)+\min_{\Omega}(\rho_0-\mu)\le \rho(t,x)\le \mu(x)+\max_{\Omega}(\rho_0-\mu).
\]
 Hence, if
\begin{align*}
a:&= \min_\Omega\mu + \min_{\Omega}(\rho_0-\mu),
\\
b:&= \max_\Omega\mu +  \max_{\Omega}(\rho_0 - \mu),
\end{align*}
then
\[
0<a\le \rho(t,x),\mu(x)\le b
\qquad\text{for all }(t,x)\in[0,T]\times\Omega.
\]
Therefore, by Lemma \ref{h_minus_one_lem},
\[
b^{-1/2}\|\rho-\mu\|_{\dot H^{-1}(\Omega)}
\le W_2(\rho,\mu)
\le a^{-1/2}\|\rho-\mu\|_{\dot H^{-1}(\Omega)}.
\]

Applying the upper bound at time \(t\) and the lower bound at time \(0\), we obtain
\begin{multline}\label{eq:W2stabilization}
W_2(\rho_t,\mu)
\le a^{-1/2}\|\rho_t-\mu\|_{\dot H^{-1}(\Omega)}
\\\le a^{-1/2}e^{-\lambda_1 t}\|\rho_0-\mu\|_{\dot H^{-1}(\Omega)}
\le \sqrt{\frac{b}{a}}\,e^{-\lambda_1 t}W_2(\rho_0,\mu).
\end{multline}
In particular,
\[
W_2(\rho_t,\mu)\to 0
\qquad \text{as }t\to\infty.
\]
\end{proof}

\noindent(\textbf{Proposition \ref{gradient_controller_uges_prop}}.) Under the hypotheses of Assumption \ref{first_assumptions}, the controller \eqref{gradient_controller_eq} renders the setpoint $\mu$ uniformly globally exponentially stable in both the $W_2$ and $L^2$ senses.
\begin{proof}
Defining the error \(e:=\rho-\mu\), and using that \(\mu\) is independent of \(t\), we obtain the heat equation
\begin{equation}\label{eq:error_heat}
  \partial_t e = \Delta e \qquad\text{in }\Omega.
\end{equation}
Let 
\[
V(t) = \frac12\int_\Omega | e(t,x) |^2 dx
\]

For sufficiently regular solutions 
\[
\frac{d}{dt} V(t) =  \int_{\Omega} e \ \partial_t e \ dx   = \int_{\Omega} e \Delta e \ dx  = - \int_{\Omega} | \nabla e|^2 \ dx
\]
where the boundary term vanishes as $\partial_n e = 0$ on the boundary surface. Hence $V$ is non-increasing and 

\[
\int_0^\infty \| \nabla e \|^2_{L^2(\Omega)} \le V(0)
\]
Next we observe that $\int e_0 dx = 0$ and $\int_{\Omega} e_t dx = 0$. Then using Poincare inequality
\begin{equation}\label{eq:poincare_inequality}
    \| e_t\|^2_{L^2} \le \frac{1}{\lambda_1}\| \nabla e_t\|^2_{L^2(\Omega)}
\end{equation}
for some $\lambda_1 >0$, we obtain 
\begin{equation}\label{eq:l2_timederivative}
\frac{d}{dt} \| e_t\|^2_{L^2} = -\| \nabla e_t\|^2_{L^2} \le - \lambda_1\| e_t\|^2_{L^2}. 
\end{equation}
Therefore, using Gr\"onwall we obtain,
\begin{equation}\label{eq:l2gronwall}
\| e_t \|_{L^2} \le \exp(-\lambda_1 t ) \| e_0\|_{L^2}
\end{equation}
and hence, $\rho_t \to \mu$ in $L^2(\Omega)$.
Hence, this establishes the UGES property of the controller w.r.t. $L^2$ norm.

The UGES w.r.t. $W_2$ is already proven in \eqref{eq:W2stabilization}. 
\end{proof}

\noindent \textbf{(Proposition \ref{gradient_controller_effort_prop}.)}
Under the hypotheses of Assumption \ref{first_assumptions}, the feedback law
\[
\K(\rho_t,\mu)=-\frac1{\rho_t}\nabla(\rho_t-\mu)
\]
has finite control effort\footnote{Note that by the uniform bounds on $\rho$, $\K$ also has finite control effort in $L^2(\rho)$, which is a relevant quantity in optimal transport and fluid mechanics.} in \(L^2(\Omega)\), namely
\[
\int_0^\infty \|\K(\rho_t,\mu)\|_{L^2(\Omega)}\,dt<\infty.
\]

\begin{proof}
Define \(e_t := \rho_t - \mu\), and note that \(e_0 \in H^1(\Omega)\) since \(\rho_0, \mu \in C^2(\Omega)\). 
Since \(e_t=\rho_t-\mu\) solves
\[
\partial_t e_t=\Delta e_t
\]
with homogeneous Neumann boundary conditions and zero mean, we have
\[
e_t=e^{t\Delta}e_0.
\]
Moreover, from the positivity estimate on \(\rho_t\),
\[
\rho_t\ge c_0>0 \qquad \text{for all } t\ge 0.
\]

Let \((\phi_k)_{k\ge 0}\) be an orthonormal basis of Neumann eigenfunctions of \(-\Delta\), with eigenvalues
\[
0=\lambda_0<\lambda_1\le \lambda_2\le \cdots.
\]
Since \(e_0\) has zero mean, its \(\phi_0\)-component vanishes, so
\[
e_0=\sum_{k\ge 1} a_k\phi_k,
\;\text{and therefore}\;
e_t=\sum_{k\ge 1} a_k e^{-\lambda_k t}\phi_k.
\]
We use this spectral decomposition of \(e_t\) to obtain
\begin{multline*}
\|\nabla e_t\|_{L^2(\Omega)}^2
=
\sum_{k\ge 1}\lambda_k a_k^2 e^{-2\lambda_k t}
\\ \le
e^{-2\lambda_1 t}\sum_{k\ge 1}\lambda_k a_k^2
=
e^{-2\lambda_1 t}\|\nabla e_0\|_{L^2(\Omega)}^2.
\end{multline*}
Hence
\[
\|\nabla e_t\|_{L^2(\Omega)}
\le e^{-\lambda_1 t}\|\nabla e_0\|_{L^2(\Omega)}
\le e^{-\lambda_1 t}\|e_0\|_{H^1(\Omega)}.
\]
This leads us to the following bound on \(\K(\rho_t,\mu)\):
\begin{align*}
\|\K(\rho_t,\mu)\|_{L^2(\Omega)}
&=
\left\|\frac1{\rho_t}\nabla e_t\right\|_{L^2(\Omega)}
\le \frac1{c_0}\|\nabla e_t\|_{L^2(\Omega)}
\\
&\le \frac1{c_0} e^{-\lambda_1 t}\|e_0\|_{H^1(\Omega)}.
\end{align*}
Integrating in time, we obtain
\[
\int_0^\infty \|\K(\rho_t,\mu)\|_{L^2(\Omega)}\,dt
\le
\frac{1}{c_0\lambda_1}\|e_0\|_{H^1(\Omega)}<\infty.
\]
This proves the claim.
\end{proof}
\noindent \textbf{(Proposition \ref{gradient_controller_symmetries_prop}.)}
Under the hypotheses of Assumption \ref{first_assumptions}, the feedback law
\[
\K(\rho_t,\mu)=-\frac1{\rho_t}\nabla(\rho_t-\mu)
\]
is symmetric, in the sense of Definition \ref{symmetry_def}, under spatial translation symmetry and spatial rotation symmetry.

\begin{proof}
To show symmetry, it suffices to verify equivariance of \(\K\) under the relevant group actions.

\noindent\emph{Spatial translations.}
For any \(z \in \mathbb{R}^d\), let
\begin{align*}
\K(g_z\rho,g_z\mu)(x)
&=
-\frac{\nabla\big((g_z\rho)-(g_z\mu)\big)(x)}{(g_z\rho)(x)} \\
&=
-\frac{\nabla(\rho-\mu)(y)}{\rho(y)} =
\K(\rho,\mu)(y).
\end{align*}
On the other hand, since the pushforward of a vector field is
\[
(g_zv)(x)=v(g_z^{-1}(x))=v(y),
\]
we have
\[
(g_z\K(\rho,\mu))(x)=\K(\rho,\mu)(y).
\]
Therefore,
\[
\K(g_z\rho,g_z\mu)(x)
=
(g_z\K(\rho,\mu))(x),
\]
so $\K$ is symmetric/equivariant with respect to spatial translations.

\noindent \emph{Spatial rotations.}
For any \(R \in SO(d)\), we compute
\begin{align*}
\K(g_{R,z}\rho,g_{R,z}\mu)(x)
&=
-\frac{\nabla\big((g_{R,z}\rho)-(g_{R,z}\mu)\big)(x)}{(g_{R,z}\rho)(x)} \\
&=
-\frac{R\,\nabla(\rho-\mu)(y)}{\rho(y)} =
R\K(\rho,\mu)(y).
\end{align*}
On the other hand, since the pushforward of a vector field is
\[
(g_{R,z}v)(x)=R\,v(g_{R,z}^{-1}(x))=R\,v(y),
\]
we have
\[
(g_{R,z}\K(\rho,\mu))(x)=R\K(\rho,\mu)(y).
\]
Therefore,
\[
\K(g_{R,z}\rho,g_{R,z}\mu)(x)
=
(g_{R,z}\K(\rho,\mu))(x),
\]
so $\K$ is symmetric/equivariant with respect to spatial rotations.
\end{proof}

\subsection{Analysis of Linearized Dynamics} \label{linearized_dynamics_appendix}

The results in Section \ref{main_results_sec} rely heavily on the analysis of the linearized dynamics \eqref{linearized_dynamics_eq}, \eqref{vector_field_b_eq} and are thus proven under the following assumptions.

\begin{assumption} \label{linearization_assumption}
    The following results in Appendices \ref{linearized_dynamics_appendix}, \ref{symmetries_appendix}, and \ref{mixing_appendix} are proven under the following assumptions:
    \begin{enumerate}
        \item $\Omega \subset \R^n$ is compact and convex with smooth boundary.
        \item $v$ satisfies the no-flux condition $v \cdot \hat{n} = 0$ on $\partial \Omega$.
        \item $\rho_t$ and $\mu$ admit densities which are $C^0$ and uniformly bounded above and below by positive constants for all time: $0 < a \leq \rho_t, \mu \leq b < \infty$.
        \item The functions $x \mapsto k(x,\rho(x),\mu(x))$ and $x \mapsto \partial_\rho k(x,\rho(x),\mu(x))$ are $C^0$ for all admissible $\rho$ and $\mu$.
        \item The function $r \mapsto k(x,r,\mu(x))$ is $C^1$ for each $x \in \Omega$.
        \item Both the nonlinear dynamics \eqref{nonlinear_dyn_eqn} and the (formal) linearized dynamics \eqref{linear_dyn_eqn} are well-posed.
        \item The flow $\phi$ of the vector field \eqref{vector_field_b_eq} is globally well-defined, and for each $t\ge 0$, the map $\phi_t:=\phi(t,\cdot):\Omega\to\Omega$ is a bi-Lipschitz homeomorphism. In particular, $D_x\phi_t(x)$ exists for a.e.\ $x\in\Omega$.
        \item $\rho = \mu$ is a fixed point of the nonlinear dynamics \eqref{nonlinear_dyn_eqn}
        \item For every $T>0$, the nonlinear dynamics \eqref{nonlinear_dyn_eqn} satisfy $\| R(e(t)) \|_{H^{-1}( \Omega)} = o(\| e(t) \|_{H^{-1}( \Omega)})$ uniformly, where $R(e)$ denotes the \emph{remainder} of the linearization as defined in \eqref{remainder_eqn}.
    \end{enumerate}
\end{assumption}

Consider the nonlinear dynamics
\begin{equation} \label{nonlinear_dyn_eqn}
    \partial_t \rho_t(x) ~=~ - \nabla \cdot \big( \rho_t(x) \, k(x,\rho_t(x),\mu_t(x))\big) ,
\end{equation}
interpreted in the weak sense. Suppose that $\rho = \mu$ is a fixed point of \eqref{nonlinear_dyn_eqn}:
\begin{equation}
    \nabla \cdot (\mu(x) \, k(x,\mu(x),\mu(x)) ~=~ 0 ,
\end{equation}
and consider the (formal) linearized dynamics around $\mu$:
\begin{equation} \label{linear_dyn_eqn}
    \partial_t e_t(x) ~=~ - \nabla \cdot ( e_t(x) \, b(x) ) ,
\end{equation}
also interpreted in the weak sense, where $e_t := \rho_t - \mu$ and
\begin{equation}
    b(x) ~:=~ k(x,\mu(x),\mu(x)) + \mu(x) \, \partial_\rho k(x,\mu(x),\mu(x)) ,
\end{equation}
where $\partial_\rho k$ denotes the partial derivative of $k$ with respect to its second argument. Define the remainder of the linearization
\begin{equation} \label{remainder_eqn}
    R(e) ~:=~ -\nabla \cdot \big[(\mu + e) \, k(\cdot, \mu+e , \mu) - e b \big] .
\end{equation}
\noindent \textbf{(Lemma \ref{linearized_dynamics_lem}.)}
Suppose Assumptions \ref{linearization_assumption} holds and that (i) both \eqref{nonlinear_dyn_eqn} and \eqref{linear_dyn_eqn} are well-posed. Then \(\mu\) is \(W_2\)-uniformly locally exponentially stable for \eqref{nonlinear_dyn_eqn} if and only if \(0\) is \(H^{-1}\)-uniformly globally exponentially stable for the linearized dynamics \eqref{linear_dyn_eqn}.

\begin{proof}
    We first consider the nonlinear $\rho$ dynamics in the weak sense:
    \begin{equation} \label{nonlin_rho_dynamics}
        \lb \partial_t f , \rho \rb + \lb \nabla f , \rho \, \K(\rho,\mu) \rb ~=~ 0
    \end{equation}
    for all $f \in C^1_c((0,T) \times \Omega)$. By defining $e := \rho - \mu$ and using the fact that $\partial_t \mu = 0$, we may write the equivalent nonlinear $e$ dynamics:
    \begin{equation} \label{nonlin_e_dynamics}
        \lb \partial_t f , e \rb + \lb \nabla f , (\mu+e) \, \K(\mu + e,\mu) \rb ~=~ 0 .
    \end{equation}
    Expanding $\K$ in $e$
    \[ \K(\mu+e,\mu) ~=~ \K(\mu,\mu) + e \partial_\rho \K(\mu,\mu) + r(e) , \]
    where $r(e) = o(e)$. Next substituting the above expression into \ref{nonlin_e_dynamics}, and using the fact that $\mu$ is a fixed point of the dynamics \ref{nonlin_rho_dynamics} (so that $\nabla \cdot (\mu \, \K(\mu , \mu)) = 0$), we obtain
    \begin{multline} \label{remainder_e_dynamics}
        \lb \partial_t f , e \rb + \lb \nabla f , e \, (\K(\mu,\mu) + \mu \partial_\rho \K(\mu,\mu)) \rb \\
        + \lb f , R(e)  \rb~=~ 0 ,
    \end{multline}
    where
    \begin{align}
        R(e) ~:=&~ - \nabla \cdot \left[ e^2 \partial_\rho \K(\mu,\mu) + (\mu+e) r(e) \right] \\
        ~=&~ - \nabla \cdot \left[ (\mu+e) \K(\mu+e,\mu) -\mu \K(\mu,\mu) - eb \right]
    \end{align}
    is the \emph{remainder} of the linearization. Define the (formal) linearized $\tilde{e}$ dynamics by
    \begin{equation}
        \lb \partial_t f , \tilde{e} \rb + \lb \nabla f , \tilde{e} \, (\K(\mu,\mu) + \mu \partial_\rho \K(\mu,\mu)) \rb ~=~ 0 ,
    \end{equation}
    equivalently,
    \begin{equation} \label{linearized_e_dynamics}
        \lb \partial_t f + \nabla f \cdot b , \tilde{e} \rb ~=~ 0 .
    \end{equation}
    Assuming both the nonlinear $e$ dynamics and linearized $\tilde{e}$ dynamics to be well-posed, the solutions of the linearized $\tilde{e}$ dynamics approximate those of the nonlinear $e$ dynamics for small $e$ if and only if $\lb f , R(e) \rb$ is asymptotically negligible compared to the rest of the equation.    
 
If
\begin{equation}
    \|R(e(t))\|_{H^{-1}(\Omega)}
    =
    o\big(\|e(t)\|_{H^{-1}(\Omega)}\big)
\end{equation}
uniformly for \(t\) in compact intervals, then the nonlinear remainder is
asymptotically negligible relative to the linearized dynamics in the
\(H^{-1}\)-topology. More precisely, using the variation-of-constants formula,
\begin{equation}
    e(t)
    =
    T(t)e_0
    +
    \int_0^t T(t-s)R(e(s))\,ds,
\end{equation}
where \(T(t)\) is the \(C_0\)-semigroup generated by the linearized dynamics
\eqref{linearized_e_dynamics}. The above remainder estimate implies that
\begin{equation}
    \sup_{0\le t\le T}
    \frac{
    \left\|
    \int_0^t T(t-s)R(e(s))\,ds
    \right\|_{H^{-1}}
    }{
    \|e_0\|_{H^{-1}}
    }
    \to 0
    \; \text{as } \|e_0\|_{H^{-1}}\to 0 .
\end{equation}
Hence the nonlinear semiflow \(S(t)\), generated by \eqref{nonlinear_dyn_eqn}, is Fréchet differentiable at \(\mu\), uniformly on compact time intervals, and its derivative is precisely the linearized semigroup:
\begin{equation}
    DS(t)[\mu]=T(t).
\end{equation}
Therefore, by \cite[Theorem 2.5]{aljamal2018linearized}, exponential stability of the linearized
semigroup \(T(t)\) implies local exponential stability of the nonlinear
dynamics in \(H^{-1}\). Conversely, if the nonlinear dynamics are locally
exponentially stable in \(H^{-1}\), then differentiating the nonlinear stability
estimate at \(\mu\), using \(DS(t)[\mu]=T(t)\), yields exponential stability of
\(T(t)\). Finally, the local equivalence of the \(W_2\) and \(H^{-1}\) metrics
from Lemma~\ref{h_minus_one_lem} transfers the \(H^{-1}\)-stability statement
to the desired \(W_2\)-stability statement.
\end{proof}

\noindent(\textbf{Proposition \ref{l1_error_prop}}.)
Let \(e_t\) solve \eqref{linearized_dynamics_eq} with \(e_0\not\equiv 0\) and \(\int_\Omega e_t\,dx=0\). Then under Assumption \ref{linearization_assumption}, for every \(T>0\), the $L^1$ norm of the error is conserved: $\| e_t \|_{L^1} =$ constant. Thus, the linearized dynamics \eqref{linearized_dynamics_eq} are neither finite-time nor asymptotically stabilizable to \(e =  0\) in $L^1(\Omega)$.

\begin{proof}
By Assumption~\ref{linearization_assumption}, for each $t\ge 0$ the map $\phi_t$ is bi-Lipschitz. Hence, by Rademacher's theorem, $\phi_t$ is differentiable for a.e.\ $x\in\Omega$. Define
\begin{equation}
\begin{aligned}
    A(t,x)&:=D_x\phi_t(x)\quad\text{for a.e.\ }x\in\Omega,\\
    J(t,x)&:=\det A(t,x).
\end{aligned}
\end{equation}
For a.e.\ $x$, the matrix $A(t,x)$ satisfies the variational equation
\[
\frac{d}{dt}A(t,x)=Db(\phi_t(x))\,A(t,x),\qquad A(0,x)=I.
\]
By Jacobi's formula,
\begin{equation}
\begin{aligned}
    \frac{d}{dt}J(t,x)
    &= \operatorname{tr}(Db(\phi_t(x)))\,J(t,x) \\
    &= (\nabla\cdot b)(\phi_t(x))\,J(t,x),
    \qquad J(0,x)=1,
\end{aligned}
\end{equation}
and therefore
\begin{equation}\label{eq:Jacobian_exp}
J(t,x)=\exp \left(\int_0^t (\nabla\cdot b)(\phi_s(x))\,ds\right)>0 , 
\quad\text{a.e.}~x\in\Omega.
\end{equation}

Along characteristics, the solution satisfies
\[
e(t,\phi_t(x))=\frac{e_0(x)}{J(t,x)}.
\]
Using the change of variables $y=\phi_t(x)$, which is valid since $\phi_t:\Omega\to\Omega$ is a bi-Lipschitz bijection, we obtain
\begin{equation}\label{eq:L1norm_preservation}
\begin{aligned}
\int_\Omega |e(t,y)|\,dy
&=\int_\Omega |e(t,\phi_t(x))|\,J(t,x)\,dx \\
&=\int_\Omega \left|\frac{e_0(x)}{J(t,x)}\right|\,J(t,x)\,dx \\
&=\int_\Omega |e_0(x)|\,dx.
\end{aligned}
\end{equation}
Hence the \(L^1\) norm is preserved along the flow, i.e.,
\[
\|e(t)\|_{L^1(\Omega)}=\|e_0\|_{L^1(\Omega)}
\qquad\text{for all }t\ge 0.
\]
Therefore, for any nonzero initial perturbation \(e_0\), the quantity \(\|e(t)\|_{L^1(\Omega)}\) cannot tend to zero. By continuity of the norm, the equilibrium \(\mu\) cannot be asymptotically stabilized in \(L^1\) for nontrivial initial perturbations.
\end{proof}
\medskip

\noindent(\textbf{Corollary \ref{lp_instability_cor}}.) Under Assumption \ref{linearization_assumption}, the linearized dynamics are not asymptotically stabilizable to $e = 0$ in any \(L^p(\Omega)\), \(1\le p\le\infty\). Thus no controller of the the form \eqref{pointwise_control_eq} can solve Problem \ref{distributed_regulation_prob} with $L^p$ uniform local exponential stability.

\begin{proof}
From \ref{l1_error_prop} we have that $e(t)$ cannot converge to $0$ in $L^1(\Omega)$ as $t\to\infty$. Since $\Omega$ has finite measure, convergence in $L^p(\Omega)$ implies convergence in $L^1(\Omega)$ for every $1\le p\le\infty$ (indeed,
$\|f\|_{L^1}\le |\Omega|^{1-1/p}\|f\|_{L^p}$), and therefore $e(t)\not\to 0$ in $L^p(\Omega)$ for any $1\le p\le\infty$. Since the linearized dynamics cannot be $L^p$-ULES, neither can the nonlinear dynamics.
\end{proof}

\subsection{Controller Symmetries} \label{symmetries_appendix}

In this paper, we consider only symmetries of controllers of order at most 1, i.e., of the form
\begin{equation} \label{k_def}
    \K(\rho,\mu)(x) ~:=~ k(x,\rho(x),\mu(x),\nabla\rho(x),\nabla \mu(x)) .
\end{equation}
For ease of readability, we will henceforth drop the ${(\cdot)}_*$ from the pushforward action $g_*$, but the reader is advised to keep the definitions from Section \ref{desirable_properties_sec} in mind.

We now discuss two examples which appear in the paper.

\begin{exmp}[Spatial Translations] \label{translations_exmp}
    Let $g_z$ denote spatial translation by $z \in \R^n$, that is, $g_z : x \mapsto x+z$. Then $g_z$ acts on scalar functions by $(g_z f)(x) = f(x-z)$ and vector functions by $(g_z v)(x) = v(x-z)$ since the differential $Dg_z$ is the identity. Then a controller $\K$ is symmetric with respect to spatial translations if $g_z (\K(\rho,\mu)) = \K(g_z\rho,g_z \mu)$ for all $z$, $\rho$, $\mu$, or, recognizing $\K(\rho,\mu)$ as a vector field and evaluating at $x$, if $\K(\rho,\mu)(x-z) = \K(g_z \rho,g_z \mu)(x)$. Applying the form \eqref{k_def}, this becomes
    \begin{multline}
        k(x-z,\rho(x-z),\mu(x-z),\nabla\rho(x-z),\nabla \mu(x-z)) \\
        ~=~ k(x,\rho(x-z),\mu(x-z),\nabla\rho(x-z),\nabla \mu(x-z)) ,
    \end{multline}
    where we have used the fact that
    \begin{equation}
        \nabla (g_z \rho)(x) ~=~ \nabla \rho (I-z)(x) ~=~ \nabla \rho(x-z)
    \end{equation}
    and similarly for $\mu$. We can therefore see that the above equation holds for all $z$, $\rho$, $\mu$ if and only if $k$ is independent of $x$. Thus, a controller $\K$ of the form \eqref{k_def} is symmetric/equivariant with respect to spatial translations if and only if is independent of $x$.
\end{exmp}

\begin{exmp}[Spatial Rotations] \label{rotations_exmp}
    Let $g_{R,z}$ denote spatial rotation by the rotation matrix $R \in SO(n)$ about the point $z$, that is, $g_{R,z} : x \mapsto R(x-z) + z$. Then $g_{R,z}$ acts on scalar functions by $(g_{R,z}f)(x) = f(R^{-1}(x-z)+z)$ and vector functions by $(g_{R,z}v)(x) = Rv(R^{-1}(x-z)+z)$. Then applying the definitions and the form \eqref{k_def}, a controller $\K$ is symmetric with respect to spatial rotations if
    \begin{multline}\label{eq:rotation_equivariance}
        Rk(y,\rho(y),\mu(y),\nabla \rho(y), \nabla \mu(y)) \\
        ~=~ k(x,\rho(y),\mu(y),R \nabla \rho(y), R \nabla \mu(y))
    \end{multline}
    for all $R$, $y := R^{-1}(x-z)+z$, $\rho$, $\mu$, and where we have used the fact that
    \begin{multline}
        \nabla (g_{R,z} \rho)(x) ~=~ \nabla (\rho(R^{-1}(I-z)+z))(x) \\
        ~=~ (R^{-1})^T(\nabla \rho)(y) ~=~ R \nabla \rho(y)
    \end{multline}
    since $R$ is orthogonal (and similarly for $\mu$).
\end{exmp}

\noindent{\textbf{(Proposition \ref{rotational_controller_prop}.)}} The only controller of the form \eqref{pointwise_control_eq} which is symmetric under spatial rotations is the trivial controller
    \begin{equation}
        k(x,\rho(x),\mu(x)) ~=~ 0 .
    \end{equation}

\begin{proof}
Using \eqref{eq:rotation_equivariance} we have
\[
k(x,r,m)=R\,k(x,r,m).
\]
Thus $k(x,r,m)$ is fixed by every rotation $R\in SO(d)$ that fixes $x$.  The only vector in $\mathbb R^d$ fixed by all rotations is the zero vector. Therefore
\[
k(x,r,m)=0.
\]
Since $x,r,m$ were arbitrary, it follows that $k\equiv 0$.
\end{proof}

\noindent (\textbf{Corollary \ref{rotational_controller_cor}.}) No controller of the form \eqref{pointwise_control_eq} which is symmetric under spatial rotations solves Problem \ref{distributed_regulation_prob}.

\begin{proof}
    Since the only controller of the form \eqref{pointwise_control_eq} which is symmetric under spatial rotations is the trivial controller $\K(\rho,\mu)=0$, each evolution under this controller is constant in time: $\rho_t = \rho_0$. Thus the dynamics cannot be attracting.
\end{proof}

\noindent (\textbf{Proposition \ref{translational_controller_prop}.}) Under Assumption \ref{linearization_assumption}, no controller of the form \eqref{pointwise_control_eq} which is symmetric under spatial translations solves Problem \ref{distributed_regulation_prob} with $W_2$ uniform local exponential stability.

\begin{proof}
    If $\K$ is a controller of the form \eqref{pointwise_control_eq} which is symmetric under spatial translations, then by Example \ref{translations_exmp}, $k$ is independent of $x$. We may then construct a counterexample to stability by considering a setpoint $\mu$ which is constant. Under Assumption \ref{linearization_assumption}, Lemma \ref{linearized_dynamics_lem} applies, and the linearized dynamics take the form \eqref{linear_dyn_eqn} for a vector field $b$ which is constant, since $\mu$ is constant, and $k$ is independent of $x$. By the no-flux boundary condition on $\Omega$, the vector field $b$ must therefore be 0. Thus $e$ must be constant-in-time, the linearization cannot be $H^{-1}$-UGES, and the nonlinear dynamics cannot be $W_2$-ULES.
\end{proof}

\subsection{Mixing Dynamics} \label{mixing_appendix}

\noindent{\textbf{(Theorem \ref{mixing_thm}.)}} Under Assumption \ref{linearization_assumption}, the following are equivalent:
\begin{enumerate}
    \item Zero is a weakly attractive fixed point of the linearized dynamics
    \begin{equation} \label{lin_dyn_mixing}
        \partial_t e_t ~=~ -\nabla \cdot (e_t b)
    \end{equation}
    starting from initial zero-mean signed measures $e_0$ absolutely continuous w/r/t Lebesgue.
    \item There exists a unique (weakly) limiting invariant measure $\pi$ among all evolutions
    \begin{equation} \label{nu_eqn}
        \partial_t \nu_t ~=~ - \nabla \cdot(\nu_t b)
    \end{equation}
     starting from initial probability measures $\nu_0$ absolutely continuous w/r/t Lebesgue.
\end{enumerate}
Furthermore,
\begin{enumerate}
    \renewcommand{\theenumi}{\roman{enumi}}%
    \renewcommand{\labelenumi}{\theenumi)}%
    \item If (1) and (2) hold and $\pi$ is absolutely continuous w/r/t Lebesgue, then $b$ must be mixing w/r/t $\pi$.
    \item If $b$ is mixing w/r/t some invariant measure $\pi'$ which dominates Lebesgue, then (1) and (2) hold and $\pi = \pi'$.
\end{enumerate}

\begin{proof}
    (1) $\Rightarrow$ (2). Starting from any initial absolutely continuous probability measure $\nu_0$, consider its evolution $\nu_t$ under the dynamics \ref{nu_eqn}. Since $\Omega$ is bounded, Prokhorov's theorem gives the existence of a weak limit point of $\nu_t$ -- call it $\pi$. The attractivity of zero in the linearized dynamics \eqref{lin_dyn_mixing} then allows us to conclude that $\pi$ is the unique limit for all absolutely continuous initial probability measures: consider the evolution $\nu_t'$ starting from any other initial absolutely continuous probability measure $\nu_0'$. Since $e_0 := \nu_0 - \nu_0'$ is zero mean and absolutely continuous, by weak global attractivity, we have
    \begin{equation} \label{mixing_proof_eqn_1}
    \begin{split}
        0 ~&=~ \lim \lb f , e_t \rb ~=~ \lim \lb f , \nu_t - \nu_t' \rb \\
        ~&=~ \lim \lb f , \nu_t \rb - \lim \lb f , \nu_t' \rb \\
        &\Leftrightarrow \quad \lim \lb f , \nu_t \rb ~=~ \lim \lb f , \nu_t' \rb ~~ \forall f \\
        &\Leftrightarrow \quad \lim \nu_t ~=~ \lim \nu_t' ~=~ \pi .
    \end{split}
    \end{equation}
    By Assumption \ref{linearization_assumption}, the flow map $\phi$ of the vector field $b$ is globally well-defined and continuous, so $\nu_t$ admits the representation $(\phi_t)_\# \nu_0$, the semigroup $(\phi_t)_\#$ is continuous, and $\pi$ must be invariant.

    (2) $\Rightarrow$ (1). Since $e_0$ is absolutely continuous and zero-mean, it may be expressed as a scaled difference of absolutely continuous probability measures: $e_0 = a(\nu_0 - \nu_0')$ for some $a \geq 0$, $\nu_0$, $\nu_0'$. Since $\pi$ is the unique limiting invariant measure, we have $\lim \nu_t = \lim \nu_t' = \pi$. Then the argument \eqref{mixing_proof_eqn_1} runs in reverse (modulo the constant $a$), and $e_t \to 0$ weakly. Since $e_0$ was arbitrary, (1) follows.

    (i). Supposing that (1) and (2) hold and $\pi \ll dx$, $\pi$ has a density w/r/t $dx$ -- call it $p$. Then breaking the function $g$ into its zero-mean and constant components w/r/t $\pi$, we have $g = \bar{g} + \tilde{g}$, where $\bar{g} := \lb 1 , g \rb_\pi$ and $\tilde{g}:= g - \bar{g}$. Note that $\lb 1 , \tilde{g} \rb_\pi = \lb 1 , \tilde{g} p \rb = 0$. Since $\pi$ is invariant, we have $(\phi_t)_\# (\tilde{g}p) = \tilde{g}_t p$, where $g_t := g \circ \phi_t^{-1}$ and similarly for $\bar{g}_t$ and $\tilde{g}_t$. Since $\bar{g}_t$ is constant, it is invariant under the flow, and since $\tilde{g}_t p = (\phi_t)_\# (\tilde{g} p)$ is zero-mean, we have $\lb f , \tilde{g}_t \rb_\pi = \lb f , \tilde{g}_t p \rb \to 0$ by weak attractivity. Putting these facts together, we have
    \begin{equation}
        \begin{split}
            \lb f, g_t \rb_\pi ~=~ \lb f , \bar{g}_t + \tilde{g}_t \rb_\pi ~=~ \lb f , \bar{g}_t \rb_\pi + \lb f , \tilde{g}_t \rb_\pi \\
            ~\to~ \lb f , \bar{g}_t \rb_\pi + 0 ~=~ \lb f , \bar{g} \rb_\pi ~=~ \lb 1 , f \rb_\pi \bar{g} \\
            ~=~ \lb 1 , f \rb_\pi \lb 1 , g \rb_\pi ,
        \end{split}
    \end{equation}
    and thus $b$ is mixing w/r/t $\pi$ by definition.

    (ii). Suppose that $b$ is mixing w/r/t some invariant measure $\pi'$ which dominates Lebesgue. Then consider any zero-mean $e_0$ which is absolutely continuous w/r/t Lebesgue. Since $e_0 \ll dx \ll \pi'$ implies $e_0 \ll \pi'$, $e_0$ admits a density w/r/t $\pi'$ -- call it $\eta_0$. Note that $\lb 1 , \eta_0 \rb_{\pi'} = \lb 1 , e_0 \rb = 0$. Then since $b$ is mixing w/r/t $\pi'$, we have
    \begin{equation}
        \begin{split}
            \lb f , e_t \rb ~=~ \lb f , \eta_t \rb_{\pi'} ~\to~ \lb 1 , f \rb_{\pi'} \lb 1 , \eta_0 \rb_{\pi'} \\
            ~=~ \lb 1 , f \rb_{\pi'} \cdot 0 ~=~ 0 ,
        \end{split}
    \end{equation}
    and since $f$ and $e_0$ were arbitrary, zero is a weakly attractive fixed point for \eqref{lin_dyn_mixing}. Now, by the equivalence of (1) and (2), there exists some invariant measure $\pi$ which is the unique limiting measure starting from all initial absolutely continuous probability measures, and we claim that $\pi = \pi'$. Consider any absolutely continuous initial probability measure $\nu_0$ and its evolution $\nu_t$. Since $\nu \ll dx \ll \pi'$, $\nu \ll \pi'$, and $\nu$ admits a density w/r/t $\pi'$, call it $n$. Note that $\lb 1 , \nu_0 \rb = \lb 1 , n_0 \rb_{\pi'} = 1$. Since $b$ is mixing w/r/t $\pi'$, we have
    \begin{equation}
        \begin{split}
            \lb f , \nu_t \rb ~=~ \lb f , n_t \rb_{\pi'} ~\to~ \lb 1 , f \rb_{\pi'} \lb 1 , n_0 \rb_{\pi'} \\
            ~=~ \lb f , 1 \rb_{\pi'} \cdot 1 ~=~ \lb f , \pi' \rb ,
        \end{split}
    \end{equation}
    thus $\lb f , \nu_t \rb \to \lb f , \pi' \rb$ for every $f$, so $\nu_t \to \pi'$. But $\nu_t \to \pi$ by definition, and so $\pi = \pi'$.
\end{proof}


\bibliographystyle{IEEEtran}
\bibliography{library}

\vspace{-5mm}

\begin{IEEEbiography}[{\includegraphics[width=1in,height=1.25in,clip,keepaspectratio]{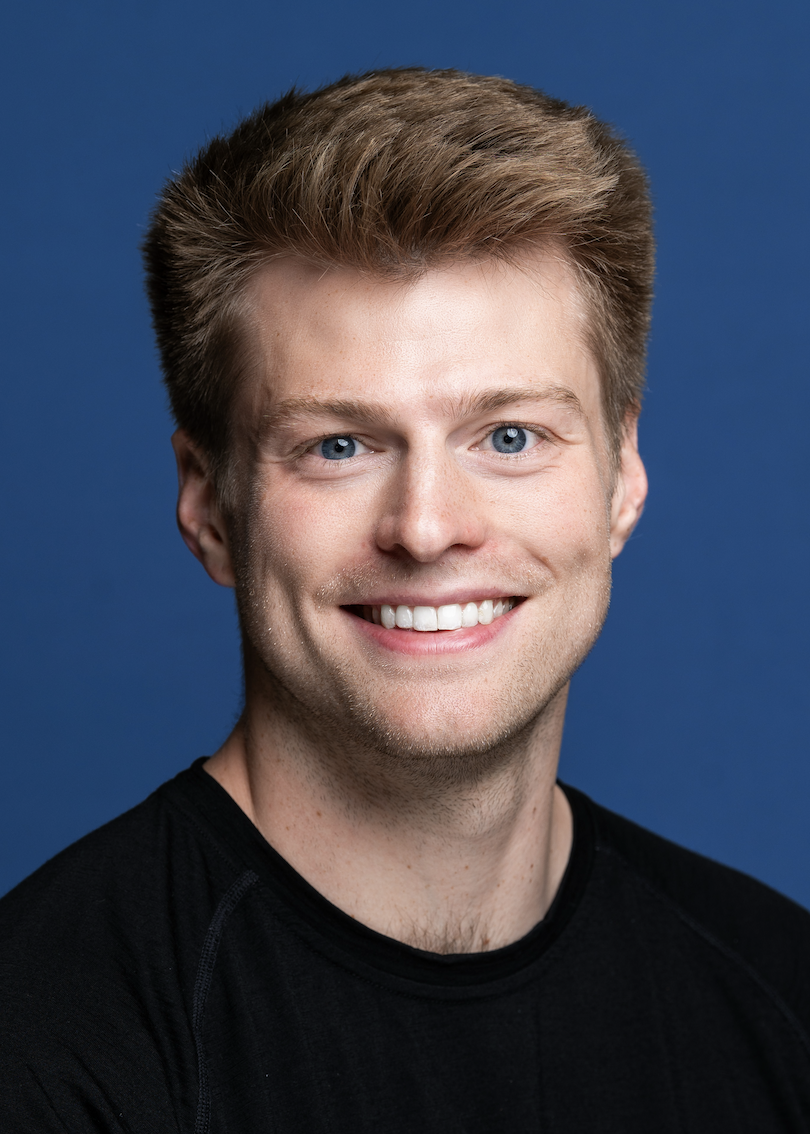}}]{Max Emerick} (GS'22)
	received the B.S. degree in mechanical engineering from California Polytechnic State University San Luis Obispo in 2020 and the M.S. degree in mechanical engineering from the University of California Santa Barbara in 2022. He is currently pursuing the Ph.D. degree in mechanical engineering at the University of California Santa Barbara.
	His research interests include optimal control and distributed parameter systems.
	Mr. Emerick is a recipient of the IFAC Young Author Award, the Conference on Decision and Control Outstanding Student Paper Award, and the National Defense Science and Engineering Graduate (NDSEG) Fellowship.
\end{IEEEbiography}

\vspace{-5mm}

\begin{IEEEbiography}[{\includegraphics[width=1in,height=1.25in,clip,keepaspectratio]{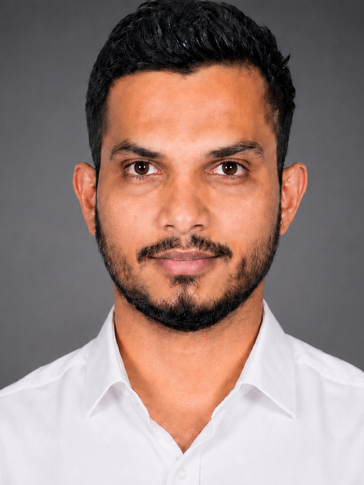}}]{Saroj Prasad Chhatoi} 
received his M.Sc. degree in Physics from the Institute of Mathematical Sciences, Chennai, India, in 2019, and his Ph.D. degree in Systems and Control from the University of Toulouse, France, in 2026. Between 2019 and 2022, he worked as a research engineer in the robotics industry and later as a research fellow at Centro E. Piaggio, University of Pisa, Italy, where he worked on optimal control for legged robots. He is currently a postdoctoral researcher in the Applied Mathematics group at TU Berlin, Germany. His research interests include stochastic and optimal control, polynomial optimization, optimal transport, and generative modeling.
\end{IEEEbiography}

\vspace{-5mm}

\begin{IEEEbiography}[{\includegraphics[width=1in,height=1.25in,clip,keepaspectratio]{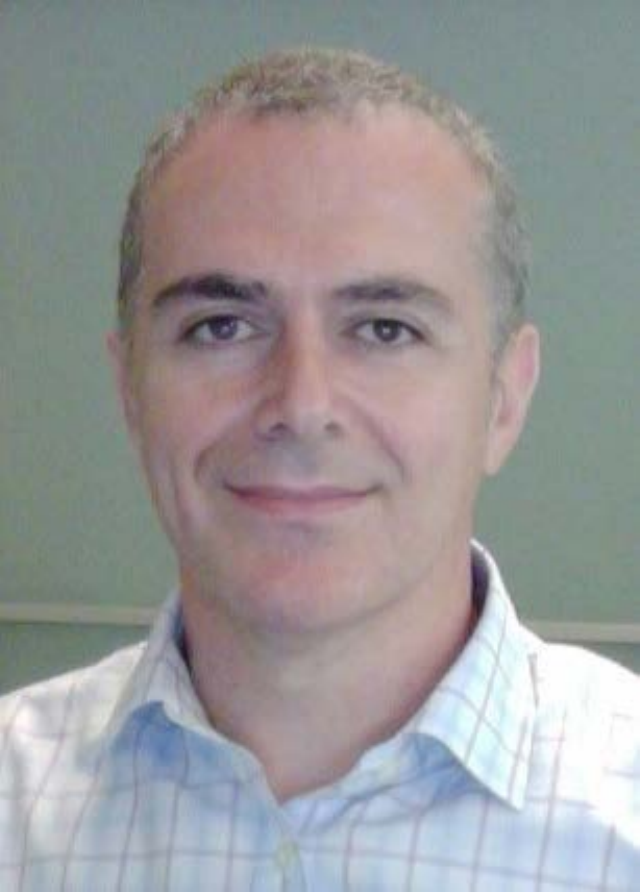}}]{Bassam Bamieh}
(Fellow, IEEE) received the B.Sc. degree in electrical engineering and physics from Valparaiso University, Valparaiso, IN, USA, in 1983, and the M.Sc. and Ph.D. degrees in electrical and computer engineering from Rice University, Houston, TX, USA, in 1986 and 1992, respectively.
From 1991 to 1998, he was an Assistant Professor with the Department of Electrical and Computer Engineering, and the Coordinated Science Laboratory, University of Illinois
at Urbana-Champaign, after which he joined the University of California at Santa Barbara (UCSB), where he is currently a Professor of Mechanical Engineering. His research interests include robust and optimal control, distributed and networked control and dynamical systems, and related problem in fluid and statistical mechanics and thermoacoustics.
Dr. Bamieh is a past recipient of the IEEE Control Systems Society G. S. Axelby Outstanding Paper Award (twice), the AACC Hugo Schuck Best Paper Award, and the National Science Foundation CAREER Award. He was a Distinguished Lecturer of the IEEE Control Systems Society (twice), and is a Fellow of the International Federation of Automatic Control (IFAC).
\end{IEEEbiography}

\end{document}